\documentclass[aip,cha,reprint,longbibliography]{revtex4-1}

\usepackage{hyperref}
\hypersetup{
    colorlinks=True,       
    linkcolor=blue,          
    citecolor=blue,        
    filecolor=magenta,      
    urlcolor=blue           
}

\usepackage{graphicx}
\usepackage{dcolumn}
\usepackage{bm}

\usepackage[utf8]{inputenc}
\usepackage[T1]{fontenc}
\usepackage{mathptmx}
\usepackage{etoolbox}
\usepackage{braket}

\usepackage{lineno,hyperref}
\usepackage{amssymb}
\usepackage{comment}
\usepackage{amsmath}
\usepackage{amsfonts}
\usepackage{amsthm}
\usepackage{color}
\usepackage{soul,xcolor}
\usepackage{cancel}

\usepackage{epsfig}

\usepackage[ruled, lined]{algorithm2e}

\numberwithin{equation}{section}

\newtheorem{theorem}{Theorem}[section]

\newtheorem{definition}[theorem]{Definition}

\newtheorem{proposition}[theorem]{Proposition}

\newtheorem{remark}[theorem]{Remark}

\makeatletter
\def\@email#1#2{%
 \endgroup
 \patchcmd{\titleblock@produce}
  {\frontmatter@RRAPformat}
  {\frontmatter@RRAPformat{\produce@RRAP{*#1\href{mailto:#2}{#2}}}\frontmatter@RRAPformat}
  {}{}
}%
\makeatother

\begin{document}


\title[]{Time-dependent Personalized PageRank for temporal networks: discrete and continuous scales} 

\author{David Aleja}

\affiliation{Departamento de Matem\'atica Aplicada, Ciencia e Ingenier\'{\i}a de los Materiales y Tecnolog\'{\i}a Electr\'onica, Universidad Rey Juan Carlos, 28933 M\'ostoles (Madrid), Spain}

\affiliation{Laboratory of Mathematical Computation on Complex Networks and their Applications, Universidad Rey Juan Carlos, 28933 M\'ostoles (Madrid), Spain}

\affiliation{Department of Internal Medicine, University of Michigan, Ann Arbor, Michigan 48109, USA}

\affiliation{Data, Complex networks and Cybersecurity Research Institute, Universidad Rey Juan Carlos, 28028 Madrid,Spain}

\author{Julio Flores}
 \email{julio.flores@urjc.es}
 
\affiliation{Departamento de Matem\'atica Aplicada, Ciencia e Ingenier\'{\i}a de los Materiales y Tecnolog\'{\i}a Electr\'onica, Universidad Rey Juan Carlos, 28933 M\'ostoles (Madrid), Spain}

\author{Eva Primo}
\affiliation{Departamento de Matem\'atica Aplicada, Ciencia e Ingenier\'{\i}a de los Materiales y Tecnolog\'{\i}a Electr\'onica, Universidad Rey Juan Carlos, 28933 M\'ostoles (Madrid), Spain}

\affiliation{Laboratory of Mathematical Computation on Complex Networks and their Applications, Universidad Rey Juan Carlos, 28933 M\'ostoles (Madrid), Spain}

\author{Miguel Romance}

\affiliation{Departamento de Matem\'atica Aplicada, Ciencia e Ingenier\'{\i}a de los Materiales y Tecnolog\'{\i}a Electr\'onica, Universidad Rey Juan Carlos, 28933 M\'ostoles (Madrid), Spain}

\affiliation{Laboratory of Mathematical Computation on Complex Networks and their Applications, Universidad Rey Juan Carlos, 28933 M\'ostoles (Madrid), Spain}

\affiliation{Data, Complex networks and Cybersecurity Research Institute, Universidad Rey Juan Carlos, 28028 Madrid,Spain}

\date{\today}

\begin{abstract}
In this paper we explore the PageRank of temporal networks on both discrete and continuous time scales in the presence of personalization vectors that vary over time. Also the underlying interplay between the discrete and continuous settings arising from discretization is highlighted. Additionally, localization results that  set bounds to the  estimated influence of the personalization vector on  the ranking of a particular node are given. The theoretical results are illustrated by means of some real and synthetic examples.

\end{abstract}


\maketitle

\begin{quotation}
Complex networks are almost everywhere, since a huge amount of systems in nature, society, and technology can be represented as graphs. Most of complex systems can be modeled as a (finite) number of elements that interact/link between them, but in many cases these interactions won't last forever. The study of such time-varying systems is the central goal of Temporal Networks Analysis. 
In this paper we focus on the well-known concept of PageRank for a time-dependant network. More specifically we formally define the PageRank of a network that evolves on either a discrete or a continuous time scale. Additionally, the natural ingredients in the classical defintion of PageRank, namely the transition probability matrix, the damping factor and the personalization vector, are simultanously made time-dependant, which to our knowledge has not being done before in the existing literature.

\end{quotation}


\section{Introduction}\label{sec:intro}

Many real-life complex systems, ranging from the Internet and metabolism to the brain connectome or social networks, can effectively be  represented as graphs. Typically, such networks serve as the framework for a dynamic system, such as data-packet traffic on the Internet or the spread of diseases in social networks, so there is a deep interplay between the structural properties of networks and the dynamical behavior of the processes that take place in these graphs (see, for example, Refs.~[\onlinecite{Review}] and~[\onlinecite{Newman}]). Among these systems,  those whose elements and interactions change over time have links that represent the sequence of interactions between nodes along time. Some examples of such model comes from many different disciplines, including social sciences (in social networks, links represent individual social interactions  such as phone calls or emails that take place from time to time), epidemiological models (that consider contact processes that change along time) or public transportation networks, among many others (see, for example, Refs.~[\onlinecite{Holme01,Holme02, Kostakos}]).

The study of structural properties and dynamical processes that take place on temporal networks have attracted the attention of scientific community in the last years and a sound literature have been produced in order to study these systems. Surprisingly, most of the existing literature on temporal networks  does not explicitly discuss the nature of time used and only considers temporal networks such that the evolution of the structure takes place in a (finite) sequence of discrete events (discrete time-scale). But the intrinsic continuous nature of time suggests considering time-varying systems such that the variations of the interactions take place in a continuous setting (continuous time-scale) and analyzing the interplay between this continuous approach with the classic discrete time-scale case (see for example, Refs.~[\onlinecite{FloresRomance}] and~[\onlinecite{Porter}]). This comparative analysis of any structural or dynamical property of continuous and discrete time-scale networks should enrich the temporal networks knowledge, giving new insights about the  universal dichotomy between continuous and discrete formalism.  Thus our contribution aims at providing the connection between these two  approaches. 


One of the milestones in Networks Science is identifying the nodes that play a central role in a given complex network (see, for example, Refs.~[\onlinecite{Review}] and~[\onlinecite{Newman}]). From transportation and technological networks to social networks, evaluating and ranking nodes based on their significance is essential for comprehending a system's behavior and there is a plethora of mathematical measures, so called centrality measures, that give quantitative solutions to this major problem (see Ref.~[\onlinecite{Newman}]). There are (almost) as many centrality measures in Network Science literature as researches working in Complex Network's Analysis, but spectral-type centrality measures based on positive eigenvectors of some matrices associated to the graph plays a relevant role in Network Analysis.  These measures include the (classic) eigenvector centrality (see Ref.~[\onlinecite{bonacich1971factoring}]) and PageRank centrality (see Ref.~[\onlinecite{PaBrMoWi}]) that illustrate the deep connection between spectral graph theory and Networks Science. 

If we translate the major problem of Centrality Analysis to the Temporal Networks' setting, it is natural considering centrality measures on such time-varying systems that identify  the nodes that play a central role along time, both for temporal networks with continuous and discrete time-scales (see, for example, Refs.~[\onlinecite{Lerman,Pan,Rocha,Takaguchi,Williams}]). In particular, the eigenvector centrality for temporal networks was also considered in Refs.~[\onlinecite{Praprotnik,Taylor,FloresRomance}] and PageRank centrality was also studied in Ref.~[\onlinecite{Porter}]. This paper will be focused on a full temporal approach to PageRank centrality of temporal networks with either continuous or discrete time-scales. It is well known the PageRank, as the stationary state of a stochastic process, has essentially three ingredients: the random navigation on the structure of the network, the damping factor $\lambda\in (0,1)$ and a personalization vector $\mathbf{v}$ (see, for example, Ref.~[\onlinecite{LangvilleMeyer}]). Hence, it is natural to think that if we are to study the PageRank centrality of temporal networks, then the temporal nature should be expected in all its ingredients. Thus, in this paper however we will focus on the time dependence of the matrix $P_A$, the transition probability matrix (see Section \ref{sec:basic}),   the personalization vector and the damping factor. 
Notice that in Ref.~[\onlinecite{Porter}], the PageRank centrality for temporal networks was addressed, although the time dependence was only considered for the random navigation on the structure of the network, leaving the personalization vector constant along time.   But since such personalization vectors can be understood as the surfing preferences of the random walkers, it is natural consider that these surfing preferences could change with time, so in this paper a  PageRank centrality of temporal networks with a personalization vector that possibly varies along time will be considered. Furthermore, some localization results that  set bounds to the  estimated influence of the temporal-varying personalization vector on  the PageRank of any particular node are given that extends the localization results proved in Ref.~[\onlinecite{GaPeRo}] for classic complex networks. On the other hand in Refs.~[\onlinecite{Gleich2014, Rossi2012}] the authors allow the personalization vector to vary while the network structure remains constant and thus our paper somehow comprehends both approaches.

We remark that since the PageRank centrality of temporal networks will be considered for networks with continuous and discrete time-scales,  a comparative analysis should be expected. The connection between continuous and discrete approaches will be established though a sampling process that understands a temporal network with discrete time-scale as a (finite) set of snapshots of a temporal network with continuous time-scale, as it was introduced in Ref.~[\onlinecite{FloresRomance}]. Notice that taking snapshots is equivalent to assuming that the temporal network only changes from time to time in a discrete manner, or, in other words, that the continuous time-scale is being replaced with a discrete one. Thus each collection of snapshots is but a sampling of the whole process. This situation can be improved, of course, by taking  a larger number of observations and, as one would expect, the PageRank of the whole evolving structure should be obtained as a limit process of the partial information (PageRank vector at some precise instant) corresponding to the different static snapshots, as the number of those increases. 

Thus the purpose of this paper is twofold. Firstly we  provide a sound fully-temporal definition of PageRank eigenvector in a temporal network, in both discrete and continuous time-scale and show how the PageRank vector of a continuous temporal network can be nicely approximated by the PageRank vectors of discrete approximations. Secondly we highlight the relevance of considering time-varying personalization vectors, proving also some localization results that give sharp bounds to the  estimated influence of the temporal-varying personalization vector on  the PageRank of any given node. Although time-dependant damping factors are also considered the simulations reveal only a limited influence of this time dependence.

We remind the reader that in Ref.~[\onlinecite{FloresRomance}] several approaches to centralities of temporal networks were given, both in the discrete and continuous settings and the proposed definitions coherently showed that the centrality eigenvector of the discretization of a continuous temporal network could be used as a good estimation of its centrality eigenvector as a continuous temporal network, as one would expect. In this paper  we will focus just on one of those definitions, the \em local heterogeneous \em centrality (see below) which will provide the driving heuristics behind our notion a PageRank vector for continuous networks. The reason is that we are interested in giving a definition  which is coherent, and in fact generalizes, the time-decay approach in Ref.~[\onlinecite{Porter}], with the extra advantage that the personalization vector will also be allowed to vary over time.

It is important to note that although the initial approach in Ref.~[\onlinecite{Porter}] is continuous in nature as the definition of the PageRank vector involves a matrix whose entries are obtained through a differential equation, the authors quickly pass to a discrete model by taking snapshots on the time interval. Coherent with this our approach will initially be discrete and will be  shown to include the definition in Ref.~[\onlinecite{Porter}]. Later we will show that this discrete model can be seen as a discretization of a continuous one that will constitute the  definition of PageRank in a continuous temporal network. Finally we will show that in this context a localization result in the sense above can be given. 

In order to fully show the connection of our contribution to Ref.~[\onlinecite{Porter}] some additional remarks about dangling nodes are given.
Notice that the existence of dangling nodes  prevents from obtaining stochasticity via normalization in matrix $A$. Traditionally this situation was solved by adding some teleportation probability vector that allows to randomly jump from a dangling node as in Ref.~[\onlinecite{LangvilleMeyer}], and this is the approach followed in our paper. We might call this a \em teleportation PageRank \em model. However in Ref.~[\onlinecite{Porter}] a somehow different  approach is taken in as far as no teleportation vector is added but the Monroe Penrose's pseudo-inverse matrix of the matrix $A$ is considered instead. This gives rise to what might be called a \em pseudo PageRank \em model. The good thing is that both approaches are intimately  connected  in the sense that the normalization of the solution vector to the pseudo PageRank problem is in fact the solution to the teleportation PageRank problem (Refs.~[\onlinecite{DelCorsoGulliRomani}] and~[\onlinecite{ArrigoHigham}]). On the other hand it is evident that in our paper all the techniques employed concerning convergence from the discrete to the continuous model refer to standard measure theory and hold independently of which of the approaches above is taken. Hence although in what follows beyond definition no dangling nodes are considered the explanation above justifies that the material in Ref.~[\onlinecite{Porter}] fits into our contribution as a particular case, which still is very interesting on its own.

The structure of the paper is as follows. In Section~\ref{sec:basic} we recall the pertinent definitions including the detailed description of  the Google matrix, and the heuristics behind, associated to the  temporal network whose eigenvector will be the sought PageRank. Section \ref{sec:discrete} deals with discrete time scales  and highlights by means of several examples the role that a time-varying personalization vector has over the PageRank. Section \ref{sec:continuousTime} deals with continuous time scales focusing on the role that the PageRank vectors of the discretizations play in approximating the PageRank vector of the continuous network.
In both Section \ref{sec:discrete} and \ref{sec:continuousTime} the problem of localization is considered. It refers to delimiting the maximum effect that the choice of a personalization vector has over the PageRank of the network.

\section{Notation and some preliminary definitions}\label{sec:basic}

We recall some notation from Refs.~[\onlinecite{GaPeRo}] and~[\onlinecite{Pe16}]. Vectors of $\mathbb{R}^{n \times 1}$ will be denoted by column matrices and we will use the superscript $T$ to indicate matrix transposition. The vector of $\mathbb{R}^{n \times 1}$ with all its components equal to $1$ will be denoted by $\mathbf{e}$. That is, $\mathbf{e}=(1,\cdots,1)^T$.

Let ${\mathcal G}= (V, E)$ be a directed  graph where $V = \{ 1,2, \ldots, n \}$ is the set of nodes and $n\in \mathbb{N}$. The pair $(i,j)$ belongs to the set  $E$ if and only if there exists a link connecting node $i$ to node $j$. The {\it adjacency matrix} of ${\mathcal G}$ is an $n\times n$-matrix
\[
A=(a_{ij})  \hbox{ where } a_{ij}=\left\{
       \begin{array}{ll}
         1, & \hbox{if $(i,j)$ is a link of ${\mathcal G}$}, \\
         0, & \hbox{otherwise.}
       \end{array}\right.
\]
A link $(i,j)$ is said to be an {\em outlink} for node $i$ and an {\em inlink} for node $j$. We denote $k_{out}(i)$ the {\it outdegree} of node $i$, i.e,  the number of outlinks of a node $i$. Notice that $k_{out}(i)=\sum_k a_{ik}$. The graph ${\mathcal G}= (V,E)$ may have {\it dangling nodes}, which are nodes $i\in V$ with zero outdegree. Dangling nodes are characterized by a vector $\mathbf{ d}\in \mathbb{R}^{n \times 1}$  with components $d_i$ defined by
\[
d_i=\left\{
       \begin{array}{ll}
         1, & \hbox{if $i$ is a dangling node of ${\mathcal G}$}, \\
         0, & \hbox{otherwise.}
       \end{array}
     \right.
\]

Notice that in what follows there is no loss of generality in assuming that the graph is weighted in which case $0\le a_{ij}$. This will be particularly relevant in the case or continuous networks in Section \ref{sec:continuousTime}.

Let  $P_A=(p_{ij}) \in\mathbb{R}^{n\times n}$  be the {\it row stochastic matrix} associated to ${\mathcal G}$ defined in the following way:
\begin{itemize}
\item if $i$ is a dangling node, $p_{ij}=0$ for all $j=1,\dots, n$,
\item otherwise, $p_{ij}=\frac{a_{ij}}{k_{out}(i)}=\frac{a_{ij}}{\sum_{k}a_{ik}}$.
\end{itemize}

Note that each coefficient $p_{ij}$ can be considered as the probability of moving from the node $i$ to the node $j$.

We recall that one of the features of the personalized PageRank algorithm is that some {\em extra} probability of jumping is given to any node. This extra or {\em teleportation} probability is assigned by using a {\em personalization vector} ${\bf v}$, which is a probability distribution vector. If, in addition, the graph has dangling nodes then the algorithm needs to assign an additional probability of jumping from these dangling nodes; this is done by introducing a probability distribution vector {\bf u}. With these ingredients, plus a teleportation parameter $\lambda$, we have everything needed to build a primitive and stochastic matrix, called {\it Google matrix}, that we denote by $G$.

Formally, $G=G(\lambda,\mathbf{u},\mathbf{v})$, with $\lambda \in (0,1)$, is defined as
\begin{equation}\label{e:PRMatrix}
G=\lambda(P_A+\mathbf{d}\mathbf{u}^T)+(1-\lambda)\mathbf{e}\mathbf{v}^T\in \mathbb{R}^{n\times n}.
\end{equation}
Note that $G$ is row-stochastic, i.e., $G\mathbf{e}=\mathbf{e}$. Recall that  $\mathbf{v}\in \mathbb{R}^{n \times 1}$, with $\mathbf{v}>0$, i.e. all its entries are positive, and $\mathbf{v}^T\mathbf{e}=1$. Analogously, $\mathbf{u}\in \mathbb{R}^{n \times 1}$ such that $\mathbf{u}>0$  and $\mathbf{u}^T\mathbf{e}=1$.

The {\it PageRank vector} $\mathbf{\pi}=\mathbf{ \pi}(\lambda,\mathbf{u},\bf{v})$ is the unique positive eigenvector of $G^T$ associated to eigenvalue 1 such that $\mathbf{\pi}^T\mathbf{ e}=1$, i.e., $\mathbf{\pi}>0$, $\mathbf{\pi}^T\mathbf{e}=1$ and $\mathbf{\pi}^TG=\mathbf{\pi}^T$ (see Ref.~[\onlinecite{PaBrMoWi}]). Since we focus our interest in ${\bf v}$ we also refer to $\pi$ as the {\em personalized PageRank vector}.

Note also that from (\ref{e:PRMatrix}) we easily have
\begin{equation}\label{PRrel}
\mathbf{\pi}^T = \lambda \mathbf{\pi}^T (P_A+ \mathbf{du}^T) + (1-\lambda) \mathbf{v}^T.
\end{equation}


Although the existence of dangling nodes will not affect the results presented and everything can be adapted by simply considering the matrix $P_A+ \mathbf{du}^T$ we will, as observed in the introduction, assume the absence of dangling nodes, and thus equation above will look as
\begin{equation}
\mathbf{\pi}^T = \lambda \mathbf{\pi}^T P_A + (1-\lambda) \mathbf{v}^T.
\end{equation}

Notice that $$\pi^T(Id-\lambda P_A)=(1-\lambda)\mathbf{v}^T$$ and thus
\[
\pi^T=(1-\lambda)\mathbf{v}^T(Id-\lambda P_A)^{-1}=\mathbf{v}^TX(\lambda),
\]
where $X(\lambda)=(1-\lambda)R(\lambda)$ and $R(\lambda)=(Id-\lambda P_A)^{-1}$ is the resolvent of $P_A$ defined for all suitable $\lambda$ (see Ref.~[\onlinecite{Boldi}]).

If $X(\lambda)=(x_{ij})$, then it has been shown in Ref.~[\onlinecite{GaPeRo}] (Theorem~3.2) that
\[
\pi_i\in (\min_{j}x_{ji}, x_{ii}).
\]

Throughout the paper we will repeatedly make use of this fact. This constitutes a \em localization \em statement as it provides for every $i$ an interval  in which the $i$-coordinate of $\pi$ must be located. 

Thus, in what follows we will consider a temporal  network $(V,E(t))_{t\in I}$   with a fixed set of nodes $V=\{1,\cdots,n\}$ ($n\in\mathbb{N}$) on a  time-scale $I$   where   $A$, $\lambda$ and $\textbf{v}$ are time-dependent. Notice that $I$ can be a discrete time-set  $\{t_1,\dots,t_N\}$ or a (continuous) interval in the reals.  

In this temporal setting the associated   time-dependent row stochastic Google matrix looks like 
\[
G(t)=\lambda(t) P_A(t)+(1-\lambda(t)){\bf e}\textbf{v}^T(t).
\]
As above the matrix $(1-\lambda(t))(Id-\lambda(t) P_A(t))^{-1}$ will be denoted $X(t)$.

We make the important remark that the matrix $G(t)$ is in fact positive due to the personalization term and thus the existence of a  PageRank vector is guaranteed by the classical Perron Theorem (see, for example, Ref.~[\onlinecite{Meyer}]). Since the PageRank vector is the only positive norm one vector associated to the spectral radius we see that stochasticity gives no additional information about the node ranking already provided by the PageRank vector. However it does give  the additional information that the spectral radius must be one and this will be used in the localization results of the paper. We notice that  in Ref.~[\onlinecite{FloresRomance}], where centrality of temporal networks was considered,  the matrices involved were not positive and irreducibility had to be checked in order to use the classic  Perron-Frobenius Theorem (see, for example, Ref.~[\onlinecite{Meyer}]).  

Since we are interested in obtaining a PageRank vector that gives the node-centrality of  a temporal network, there is an obvious (or naive) approach consisting in obtaining for every $t$ the PageRank vector of $G(t)$. However we want to go deeper in the sense that the interaction of a couple of nodes $i$ and $j$ at instant $t$ (which is reflected in the matrix $P_A(t)$ by means of non-zero entry in position $ij$) should leave somehow a print (or a memory effect)  on the interaction at a later stage $s$  even if at stage $s$ there is no actual interaction (as determined by $P_A(s)$). This point of view was considered in Ref.~[\onlinecite{Porter}]  by means of a time-decay model to define those virtual interactions in terms of a simple differential equation that expresses how far back the actual interaction between the nodes took place (the farther back the weaker the print on the current stage) and thus the difference between interaction and interaction's strength naturally arises, according to the terminology used in Ref.~[\onlinecite{Porter}].

As mentioned  above, in this paper we propose a definition of a PageRank of a temporal network which somehow comprehends the approach in Ref.~[\onlinecite{Porter}]. In our  definition  the personalization vector will also be allowed to be time dependent (in contrast to the construction in Ref.~[\onlinecite{Porter}] where the personalization vector remains constant at every stage). Additionally we will let the damping factor be  time dependant as well.
Evidently one would expect that the choice of $\textbf{v}(t)$ influences the  centrality of a given node. In this paper we will also look at the largest  variation interval of the centrality of a given node for all possible personalization vectors. In other words we will obtain a \em localization \em result which delimits the maximum possible effect of the choice of the personalization vector. Thus we will establish the maximum possible variation that the ranking of a node can experiment in terms of the perturbation  caused by  the personalization vector. %

\section{Discrete time scale case}
\label{sec:discrete}

In this section we will focus on temporal networks whose nature requires rewiring to occur in a discrete manner. Hence in what follows the  time scale will be the discrete set $I=\{t_1,\dots,t_N\} \subset\mathbb{R}$ with $t_1<t_2<...<t_N$. 
In what follows our motivation stems from the work on centrality in Ref.~[\onlinecite{FloresRomance}] (see also Refs.~[\onlinecite{Sola,Taylor}]) where the leading heuristics  the centrality of node $j$ has influence on centrality of node $i$ in instant $t_k$  ($l\le k$) {\bf if} in instant $t_l$ there is a link $j\rightarrow i$.  Although an interested reader can find there all the necessary material that motivates our definitions, we will, for the sake of brevity, skip the details. 
Thus, for every $k$, the matrix $P_B(t_k)$, will be the normalization of $B(t_k)$, whose $ij$-entry is 
\[
P_B(t_k)_{ij}=\frac{\displaystyle\sum_{l\le k}\omega_{lk}A(t_l)_{ij}}{\displaystyle\sum_{p=1}^n\sum_{l\le k}\omega_{lk}A(t_l)_{ip}}.
\]
Notice that the denominator is what we called $k_{out}(i)$ above and evidently $P_B(t_k)$ is row stochastic.

The same remarks concerning dangling nodes are in order now and  no dangling nodes will be considered.  Sill it is important to observe that the accumulative character in the definition of $B$ demands for a dangling node to be relevant to remain dangling for every instant $l\le k$; otherwise the denominator of $P_B(t_k)_{ij}$ becomes different from zero and $\textbf{d}(t_k)$ and $\textbf{u}(t_k)$ are no longer needed.

\begin{remark}
The discrete case was modeled over  a finite set of integers. Notice that we could  consider instead the set of all positive integers  as the definition involves, for a given $t_k$, only the instants before and thus a finite sum is still formed. Of course the important ingredient here is the existence of a first element in the index set.
\end{remark}

\begin{definition}\label{def:local}
If $(V,E(t))_{t\in I}$ is a temporal network of $n$ nodes with discrete time-scale  $I=\{t_1,\dots,t_N\}$ and weights $\omega_{lk}\geq 0$ such that $\omega_{kk}\neq 0$ for all $k$, the {\em  
PageRank-eigenvector} at instant $t_k$  is the dominant positive eigenvector 
 $\mathbf{\pi}^T(t_k)=(\pi_1(t_k),\dots, \pi_n(t_k))\in\mathbb{R}^n$ with $\|\pi(t_k)\|_1=1$  associated to  
\begin{equation}
G(t_k)=\lambda(t_k) P_B(t_k)+(1-\lambda(t_k))\textbf{e}\textbf{v}^T(t_k).
\end{equation} 
The 
PageRank at instant $t_k$ of node $i$, $(\pi(t_k))_i$,  is defined as  $\langle\pi(t_k),\textbf{e}_i\rangle$ where $\{\textbf{e}_i\}_{i=1}^n$ is the canonical basis.
\end{definition}

Note at this point that the model given in Ref.~[\onlinecite{Porter}] becomes a particular case of our definition by simply taking 
 $\omega_{lk}=e^{-\alpha(t_k-t_l)}$  while $\alpha>0$ appears in the time-decay ODE supporting the definition. Notice also that the parameter $\alpha$ in Ref.~[\onlinecite{Porter}]  was uniformly chosen for every pair of nodes but this is clearly a restriction that could reasonably be removed. Finally notice that the personalization vector was allowed to vary with time instead of remaining constant as in Ref.~[\onlinecite{Porter}].

 \begin{figure*}
    \begin{center}
    \includegraphics[width=1\textwidth]{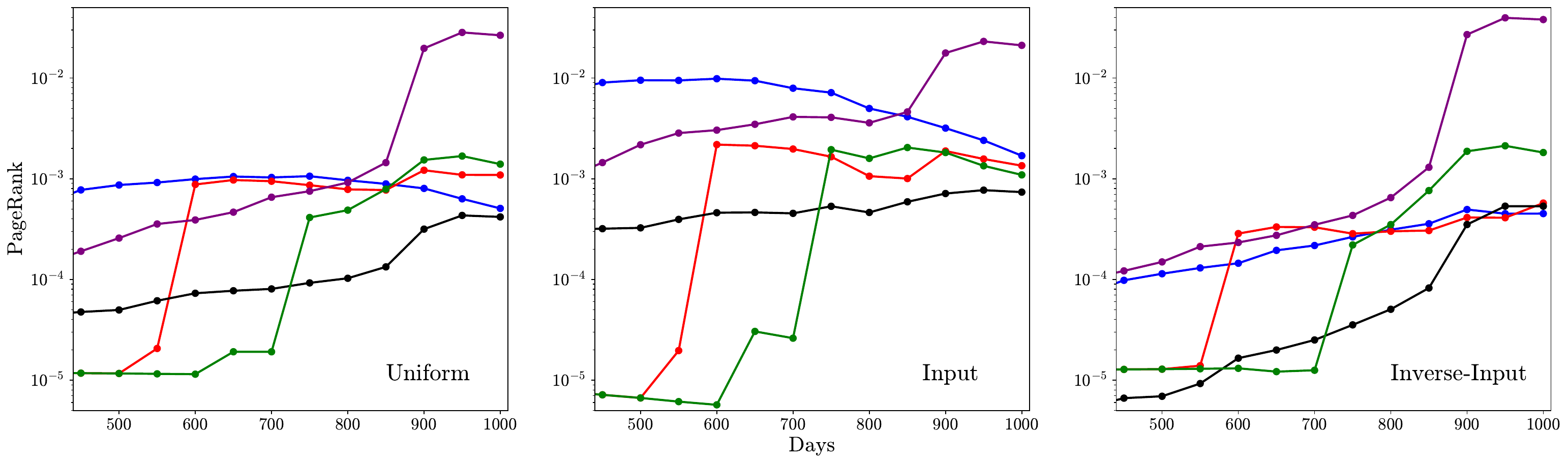}
    \end{center}
    \caption{{\bf  Effects by changing the personalization vector in Definition \ref{def:local}.} PageRank values of five nodes (a color for each node) of Italian Wikipedia temporal network (see Refs.~[\onlinecite{Jerome}] and~[\onlinecite{Preusse}]) corresponding to an uniform (left panel), Input type  (middle panel) and Inverse-Input type (right panel) personalization vector.} 
    \label{Fig1}
\end{figure*}

In order to illustrate how allowing for the personalization vector to vary over time makes our model meaningful we will consider 
 three different personalization vectors that are proportional, respectively,  to the following:

\begin{itemize}
    \item Uniform: $\mathbf{e}=(1,1, \dots, 1)$.
    \item Input: $\displaystyle \mathbf{e}+\sum_{i=1}^n(a_{i1}(t_k),a_{i2}(t_k),\dots, a_{in}(t_k))$.
    \item Inverse-Input: 
    \[
    \left(\dfrac{1}{1+\displaystyle\sum_{i=1}^na_{i1}(t_k)},\dfrac{1}{1+\displaystyle\sum_{i=1}^na_{i2}(t_k)},\dots,\dfrac{1}{1+\displaystyle\sum_{i=1}^n a_{in}(t_k)}\right).
    \]
\end{itemize}
Notice that each of those vectors makes the teleportation probability increase according to different criteria; thus while the Input-type gives special importance to the number of entry links at each stage, the Inverse-Input-type does the opposite, penalizing this number, whereas the uniform-type simply  ignores it. To see how different the results obtained are, we will consider the network of evolving hyperlinks between articles of Italian Wikipedia (see Refs.~[\onlinecite{Jerome}] and~[\onlinecite{Preusse}]). A directed link will be set between two articles in Wikipedia if there is an hyperlink in one of them that leads to the other.   The use of PageRank centrality in the analysis of Wikipedia it has been widely accepted by the scientific community since Wikipedia has become the world largest encyclopedia with open public access and its links structure reminds a structure of the World Wide Web (WWW) (see, for example, Ref.~[\onlinecite{nielsen2012wikipedia}]), so PageRank provide objective relevance scores for hyperlinked documents in Wikipedia (see Ref.~[\onlinecite{eom2015interactions}] and references therein). Furthermore, since Wikipedia is an alive on-line encyclopedia, all its entries and hyperlinks evolve each day, so it is natural analyzing as a temporal network (see, for example, Ref.~[\onlinecite{eom2013time}]).

In the reference Ref.~[\onlinecite{konect}], the temporal network is represented as follows:
- We start from the initial adjacency matrix,
- At each timestamp we add to this matrix a matrix containing at each position 0, +1 or -1, depending on whether that connection remains stable, appears or disappears.
The first timestamp considered will be labeled as day zero, then additional $1000d$ timestamps will be added, where $d$ is the number of seconds in a whole day, and finally the last timestamp will be labeled as day one thousand. The number of nodes in this period is $n=82168$ with $350755$ interactions and $38063$ disconnections. The snapshot-instant set  is $I=\{t_1,t_2,\dots,t_{21}\}$ with $t_k=50(k-1)$ days, the parameters used are $\lambda=0.85$, $\alpha=0.001$ while the different time-influence weights are obtained according to $\omega_{lk}=e^{-\alpha(t_k-t_l)}$.  Moreover, when dangling nodes exist for some $t_k$, we consider $P_B(t_k)+\textbf{d}(t_k)\textbf{v}^T(t_k)$ instead of $P_B(t_k)$ (notice that $\textbf{u}(t_k)$ has been chosen equal to $\textbf{v}(t_k)$). Then  we compute the temporal Page-Rank on $I$ for the different personalization vectors above using the algorithm in  Ref.~[\onlinecite{Aleja}] due to the  existing dangling nodes. As in  Ref.~[\onlinecite{Porter}] more than half of a link's strength in time $t_k$ is expected to linger until $t_k+\eta_{1/2}$ with $e^{-\alpha \eta_{1/2}}=0.5$, or equivalently 693 days, approximately, in this particular case.

The different PageRank values of five nodes corresponding to an uniform (left panel), Input type  (middle panel) and Inverse-Input type (right panel) personalization vector in $[t_{10},t_{21}]$ are represented in Figure~\ref{Fig1}. Nodes have been chosen where most significant changes are observed for the different customization vector choices. Notice how the ranking of a given node changes with the chosen personalization vector; for instance for $t_{21}=1000$ days, the blue node ranks fourth with the uniform vector, second with the Input-type vector and last with the Inverse-Input type vector. Observe also that the blue and purple nodes share the first position with the  Uniform and  Input-type vectors, but it is the purple node which ranks first most of the time, except for a short period when the red node takes over, if the personalization vector is of the Inverse-input type. It thus becomes evident that the choice of the personalization vector determines the temporal PageRank.

Remarkably these changes occur not only on a small set of nodes. Indeed, in  Figure~\ref{Fig2} a), a pairwise comparison of the three temporal PageRanks corresponding to the three types of personalization vectors is provided  by means of the  Kendall Tau coefficient introduced in Ref.~[\onlinecite{Kendall}]. This coefficient is 1 if  two rankings coincide (this includes ties), -1 if they are opposite  and zero if they are no correlated at all. In order to reduce the complexity we make use of the Meger Sort algorithm given in Ref.~[\onlinecite{Knight}]. This  Kendall Tau coefficient is computed for each $t_i$ for the pairs Uniform vs Input, Uniform vs Inverse-Input and Input vs Inverse-Input. Figure~\ref{Fig2} a) shows how the rankings change and in fact that the change is more evident when the Inverse-Input type vector is involved;  this might be explained by the fact that the Inverse-Input type vector goes against the network structure, in contrast with the other two types. 

\begin{figure*}
    \centering
    \includegraphics[width=\linewidth]{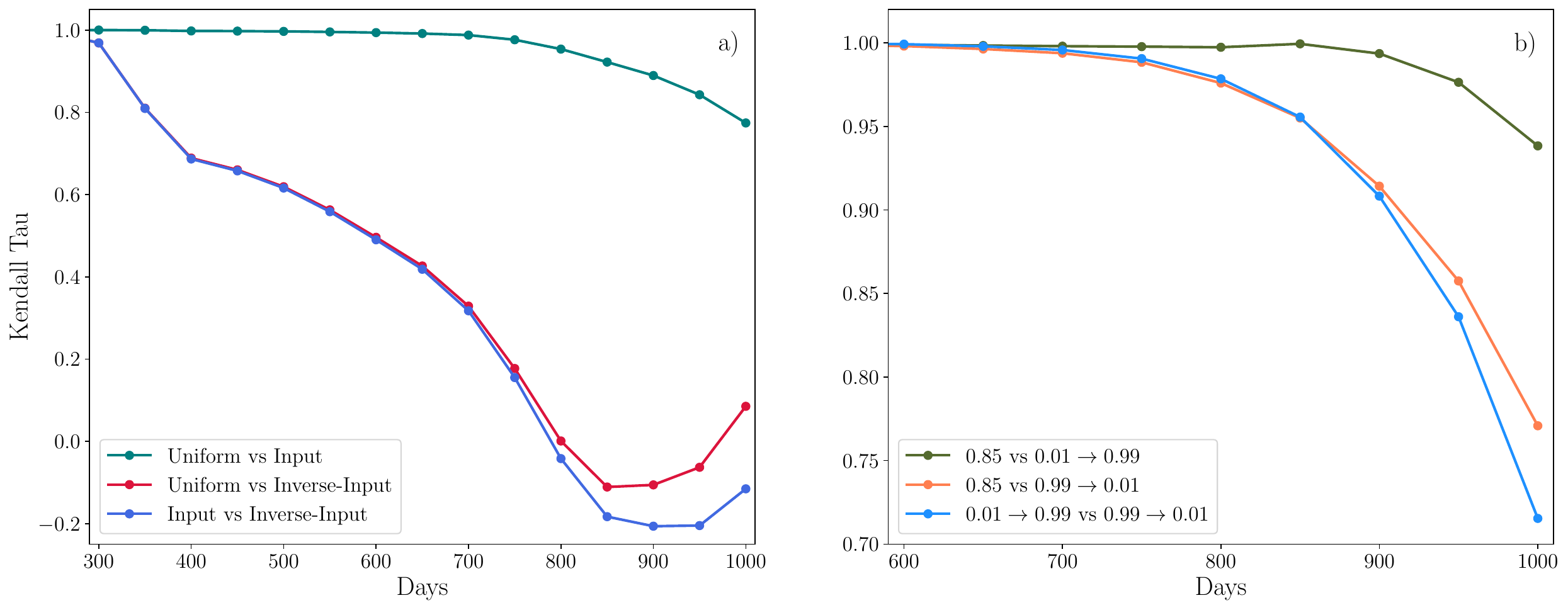}
    \caption{{\bf A comparison by changing the personalization vector or the teleportation parameter in Definition~\ref{def:local}}.  Panel a): Kendall Tau coefficients over time are shown between the three pairs of temporal PageRanks considered, with Uniform, Input and Inverse Input personalization vector. Panel b): Kendall Tau coefficients over time are shown between the three pairs of temporal PageRanks when the teleportation parameter varies.}
\label{Fig2}
\end{figure*}

Also, in Figure~\ref{Fig2} b), we show how the rankings change when the teleportation parameter varies over time. Indeed, considering the previous example with a uniform personalization vector and three different choices for the teleportation parameter: 
\begin{itemize}
    \item $0.85$: $\lambda(t_k)=0.85$,
    \item $0.01\to 0.99$: $\lambda(t_k)=0.01+(0.99-0.01)(k-1)/20$,
    \item $0.99\to 0.01$: $\lambda(t_k)=0.99+(0.01-0.99)(k-1)/20$,
\end{itemize}
in each $t_k$, $k\in \left\{1,2,\dots,21\right\}$, Figure~\ref{Fig2} b) illustrates how the Kendall Tau coefficients change over time when comparing the rankings with each other.

Although the previous examples highlight the relevance of a time-dependent personalization vector and a time-dependent teleportation parameter, the following \em localization result \em sets limits to its maximum possible effect  over $(\pi^T(t_k))_i$.

\begin{theorem} 
\label{LHPRDiscrete}
If $(V,E(t))_{t\in I}$ is a temporal network of $n$ nodes with discrete time-scale  $I=\{t_1,\dots,t_N\}$, then for every $k\in\{1,\dots, N\}$ and every node $i\in\{1,\dots, n\}$  
\[
\min_{j} (X(t_k))_{ji}\le\  ({ \pi}^T(t_k))_i\ \le \max_j (X(t_k))_{ji}=(X(t_k))_{ii},
\]
where $X(t_k)=(1-\lambda(t_k))\left(Id-\lambda(t_k) P_B(t_k)\right)^{-1}.$
\end{theorem}

\begin{proof}
By, using $\|\pi(t_k)\|_1=1$ we write
\[
\begin{split}
\pi^T(t_k)&=\pi(t_k)^T G(t_k)\\
&=\pi^T(t_k)\lambda(t_k)  P_B(t_k)+(1-\lambda(t_k) )\pi^T(t_k)\textbf{e}\textbf{v}^T(t_k)\\
&=\pi^T(t_k)\lambda(t_k)  P_B(t_k)+(1-\lambda(t_k) )\mathbf{v}^T(t_k)
\end{split}
\]
or $ \pi^T(t_k)=v^T(t_k)X(t_k)$.
 
Again, since $\|v^T(t_k)\|_1=1$ we have
\[
\begin{split}
 \langle\pi(t_k),e_i\rangle &= (v^T(t_k)X(t_k)) e_i = v^T(t_k)(X(t_k)e_i) \\
 &= \sum_{j}v^T(t_k)_jX(t_k)_{ji},
\end{split}
\]
which is a convex combination of the elements of 
$(X(t_k))_{i}$,   the $i$-th column of   $X(t_k)$. Thus
\[
\min_{j} (X(t_k))_{ji}\le\  ((\pi^T)(t_k))_i\ \le \max_j (X(t_k))_{ji}=(X(t_k))_{ii}
\]
 (see Ref.~[\onlinecite{GaPeRo}] (Theorem 3.2)).
 \end{proof}

 \section{Continuous time scale case}\label{sec:continuousTime}

Once the definition of PageRank for temporal networks on a  discrete time-scale has been established we will pass to consider temporal networks on a  {\em continuous time-scale}.  Our approach to a definition of PageRank  in this context will be based in Ref.~[\onlinecite{FloresRomance}], where local heterogeneous centrality of a continuous temporal network was given. As it was shown in  Ref.~[\onlinecite{FloresRomance}] the definition of the (continuous) PageRank will be good in as much as it relates well with the discrete one. In fact the (continuous) PageRank  will be shown to be well approximated by the (discrete) PageRank of the network obtained after  discretization of the initial network. Also, this approximation will be shown to get  sharper as the  discretization is refined. Thus the remarks made in the previous section will justify that the definition given in Ref.~[\onlinecite{Porter}] can be seen as a discretization of our definition. In what follows $I$ will denote the time interval $[t_0,t_1]\subset \mathbb{R}$.

\begin{remark}\label{piece continuity}
For a given temporal network $(V,E(t))_{t\in I}$  it is natural to assume that the function $A:I\longrightarrow M_{n\times n}$ assigning to each $t$ the adjacency matrix of $(V,E(t))$, is a piecewise continuous which is equivalent to the piecewise continuity of every entry function $a_{ij}:I\longrightarrow [0,+\infty)$. Arguably piecewise continuity is in fact the most accurate description of a real situation. However from the technical point of view working with measurable functions, which constitute a natural generalization of  piecewise continuity (see, for example, Refs.~[\onlinecite{Rudin}] and ~[\onlinecite{Cohn}]),  has some advantages and thus piecewise continuity will be replaced by measurability in the definition below. For brevity when there is need,  will say that $(V,E(t))_{t\in I}$ is measurable if all the matrix entries involved are measurable functions while $(V,E(t))_{t\in I}$ will be said  integrable if all the matrix entries  are integrable. Integration will be understood in Lebesgue's sense.
\end{remark}

 Evidently if the function $A$ above is measurable as well as the \em damping function \em $\lambda:I\rightarrow (0,1)$, so must be the  function $G:I\longrightarrow M_{n\times n}$ which assigns to every   $t\in I$  the stochastic matrix  $G(t)=\lambda(t) P_A(t)+(1-\lambda(t)){\bf e}{\bf{v}}^T(t)$. As above we make the remark that the matrix $G(t)$ is positive due to the personalization term and thus the existence of a PageRank vector is guaranteed by the classical Perron Theorem (see, for example, Ref.~[\onlinecite{Meyer}]).  

As in the discrete time setting, given the temporal network $(V,E(t))_{t\in I}$, we are interested in defining  the PageRank vector of $\{G(t)\}_{t\in I}$.

The heuristic above sustaining Definition~\ref{def:local} can easily be translated into the continuous case. Given some integrable (weight) function $\omega:I^2\longrightarrow [0,+\infty)$ satisfying $\omega(s,t)=0$ for $s> t$, we consider the matrix $B$ whose ${ij}$-entry is 
\[
B_{ij}(t)=\int_{t_0}^t \omega(s,t)a_{ij}(s) ds.
\]
As above the normalization of $B$, yields the matrix $P$ whose $i,j-$ entry is defined as 
\[
P_{ij}(t)=\frac{\displaystyle\int_{t_0}^t \omega(s,t)a_{ij}(s) ds}{\displaystyle\sum_{k=1}^n\int_{t_0}^t \omega(s,t)a_{ik}(s) ds}.
\]
As above the matrix 
\[
G(t)=\lambda(t) P(t)+(1-\lambda(t)) e v^T(t)
\] 
is formed. Notice that the remark above about dangling nodes remains pertinent in this section. In order to avoid inconsistencies in  the definition above we need to take care of  $P_{ij}(t_0)$ as the formula above evidently renders $0/0$ in $t=t_0$. We will do so by imposing the existence of some $\epsilon>0$ such that $\omega(s,t)$ is continuous in $[t_0,t_0+\epsilon]^2$ -or \em right-up \em continuous in $t_0$-and and $a_{ij}(s)$ is right continuous  in $t_0$ for all $i,j$. Notice that we said above that piece-wise continuity was the good property to impose on the functions involved regarding real situations while measurability was only introduced for technical reasons. So, our assumption is reasonable. Assuming thus that  $\omega(s,t)$ is continuous in $[t_0,t_0+\epsilon]^2$ and $a_{ij}(s)$ is continuous in $[t_0,t_0+\epsilon]$ we can, by the  Mean Value Theorem for integrals, find for every $k\in\{1,\dots, n\}$ some  $s_k\in(t_0,t)$ such that 
\[
\int_{t_0}^t \omega(s,t)a_{ij}(s)ds=\omega(s_k,t)a_{ij}(s_k)(t-t_0)
\]
and thus 
\[
\begin{split}
    P_{ij}(t)&=\frac{\omega(s_j,t)a_{ij}(s_j)(t-t_0)}{\displaystyle\sum_{k=1}^n\omega(s_k,t)a_{ik}(s_k)(t-t_0)}\\
    &=\frac{\omega(s_j,t)a_{ij}(s_k)}{\displaystyle\sum_{k=1}^n\omega(s_k,t)a_{ik}(s_k)}.
\end{split}
\]
Evidently, if $t\rightarrow t_0$ and we assume, as done in the discrete setting,  $\omega(t_0,t_0)\neq 0$ we get 
\[
\frac{\omega(s_j,t)a_{ij}(s_k)}{\displaystyle\sum_{k=1}^n\omega(s_k,t)a_{ik}(s_k)}\longrightarrow \frac{\omega(t_0,t_0)a_{ij}(t_0)}{\displaystyle\sum_{k=1}^n\omega(t_0,t_0)a_{ik}(t_0)}=p_{ij}(t_0),
\]
where $P_A(t)=(p_{ij}(t))$ was the normalization of the adjacency matrix $A(t)$ in instant $t$. This suggests to define
\[
P_{ij}(t_0)=p_{ij}(t_0).
\]

Notice that there is no need to impose $a_{ij}(t_0)\neq 0$ because the denominator becomes $k_{out}(i)$ which is never zero in the absence of dangling nodes.  All the above is collected in the following.

\begin{definition}\label{def:cont-local}
Let $I=[t_0,t_1]\subset\mathbb{R}$ and let $(V,E(t))_{t\in I}$ be an integrable temporal network  of $n$ nodes with $a_{ij}$ right continuous in $t_0$ for all $i,j$. Let $\omega:I^2\longrightarrow [0,+\infty)$ be an integrable function satisfying $\omega(s,t)=0$ for $s> t$,  $\omega(t_0,t_0)\neq 0$  and such that  it is right-up continuous in $t_0$. Let $\lambda:I\rightarrow (0,1)$ be measurable. If $G=\{G(t)\}_{t\in I}$ as above, then the {\em 
PageRank} of $G$ is a function $\pi:I\longrightarrow \mathbb{R}^n$ with $\pi_i(t)\ge 0$ for all $1\le i\le n$ and $\|\pi(t)\|_1=1$ for all  $t\in I$ such that $ \pi(t)^TG(t)=\pi(t)^T$.
\end{definition}

Observe that this definition in the continuous setting is the natural counterpart of the discrete version. 

\begin{proposition}\label{prop:local}
Let $(V,E(t))_{t\in I}$ be  an integrable temporal network  of $n$ nodes  and  $G=\{G(t)\}_{t\in I}$ as above, then:
\begin{itemize}
 \item[{\it (i)}] 
 There always exists a unique 
 PageRank. Furthermore $\pi$ is measurable.
 \item[{\it (ii)}]
 $\pi(t)^T=v^T(t)X(t)$
where 
\[
X(t)=(1-\lambda(t))\left(Id-\lambda(t) P_B(t)\right)^{-1}.
\]  
\end{itemize}
\end{proposition}

\begin{proof}
Evidently  $G(t)$ is stochastic for every $t$. Also the existence and uniqueness of the local heterogeneous PageRank stems again from the classic Perron theorem while the measurability of $\pi$ comes from Lusin's Theorem (see, for example, Ref.~[\onlinecite{Rudin}] for Lusin's Theorem and Ref.~[\onlinecite{FloresRomance}] for more details). This justifies {\it(i)}.  As for {\it (ii)} we only need to observe once again the following straightforward verification
\[
\begin{split}
    \pi^T(t)&=\pi^T(t)G(t)=\pi^T(t)\lambda(t) P(t)+(1-\lambda(t))\pi^T e v^T(t)\\
    &=\pi^T(t)\lambda(t) P(t)+(1-\lambda(t)) v^T(t).
\end{split}
\]
Hence $\pi(t)^T=v^T(t)(1-\lambda(t))\left(Id-\lambda(t) P(t)\right)^{-1}$.
\end{proof}

Theorem~\ref{LHPRDiscrete} provides in the discrete case a localization result which gives the maximum variation that the PageRank coordinates experiment no matter what personalization vector is used, even if it evolves with time.  
In this section we will derive an analogous localization result when time runs in an interval of the reals $I=[t_0, t_1]\subset\mathbb{R}$.

The approach, as one would expect, consists of  taking a sample of   finitely many instants $I_N=\{s_k\}_{k=0}^N\subset I$  to obtain  a set of snapshots $\{G(t)\}_{t\in I_N}$ of the  network $G$ which becomes now a discrete temporal network $G^N$ on a the discrete time-scale $I_N$. The PageRank vector of this discrete network will be given a maximum variation interval for each of its coordinates according to the localization result above. A pass to the limit argument together with  a continuity argument will provide then a maximum variation interval for the coordinates of the PageRank vector of our continuous temporal network. In this passing  to the limit  the diameter of $I_N$, that is the maximum distance between two consecutive instants in $I_N$, will be required to  tend to zero as $N$ grows.

This approach was already used in Ref.~[\onlinecite{FloresRomance}] when the  centrality of continuous temporal network was obtained as limit of  the centralities of its snapshots. 

Thus, in what follows $I_N$ will mean a uniform partition of  $I$, i.e. a finite subset $I_N=\{s_k\}_{k=1}^N$ of $I$ with {\em equidistant} subdivision points. The diameter of the partition is the biggest length of the subintervals in $I$ determined by the set $I_N$. For practical reasons the partition $I_{N+1}$ will also be a uniform partition which strictly refines $I_N$, i.e it includes all subdivision points of $I_N$ as a proper subset of its points. Notice that the diameter of $I_N$  surely tends to zero,  as $N$ increases.

\begin{definition}\label{def:truncation}
Let  $(V,E(t))_{t\in I}$ be a continuous temporal network  with $n$ nodes and let   $I_N=\{s_{k}\}_{k=1}^N$  be a partition of $I$. The discrete network  $(V, E(s_k))_{{s_k\in I_N}}$ is called the {\sl truncation} of $(V,E(t))_{t\in I}$ corresponding to   $\{s_{k}\}_{k=1}^N$.
\end{definition}

For each $N$ we consider $(V, E(s_k))_{{s_k\in I_N}}$,  the truncation of $(V,E(t))_{t\in I}$ and its  associated PageRank vector $\pi_N$ obtained in  Section~\ref{sec:discrete}.  The matrix $(1-\lambda(s_k))(Id-\lambda(s_k) P_B)^{-1}$ associated to $(V, E(s_k))_{s_k\in I_N}$ will be denoted $X_N$  while $X$ will still refer to the one corresponding to $(V,E(t))_{t\in I}$.

The following result is the main ingredient to make our approach meaningful. The  proof is based on the well-known fact that a (Lebesgue) measurable function can be nicely approximated by a sequence of simple functions (the samplings).  The technical details can be found in Ref.~[\onlinecite{FloresRomance}] with the evident adaptations.

The next theorem connects both continuous and discrete settings. For technical reasons the natural identification of the  PageRank $\pi_N$ vector with the measurable step function, also denoted $\pi_N$, on the interval $I$ which assigns the value $\pi_N(s_k)$ to all points in the interval 
$[s_{k},s_{k+1})$ will be made without further comments. 
\begin{theorem} 
Given an integrable temporal network $(V,E(t))_{t\in I}$ of $n$ nodes  and  $G=\{G(t)\}_{t\in I}$ as above, let $\pi$, denote the PageRank  vector of $G$ and $\pi_N$ the PageRank vector of the truncation $(V, E(s_k))_{{k=1\dots N}}$ where the diameter of the associate partition goes to zero as $N$ increases. Then 
\[
\pi(t)=\lim _{N\rightarrow \infty}\pi_N(t),
\]for all $t\in I$.
\end{theorem}

 \begin{figure*}
    \centering
    \includegraphics[width=1\linewidth]{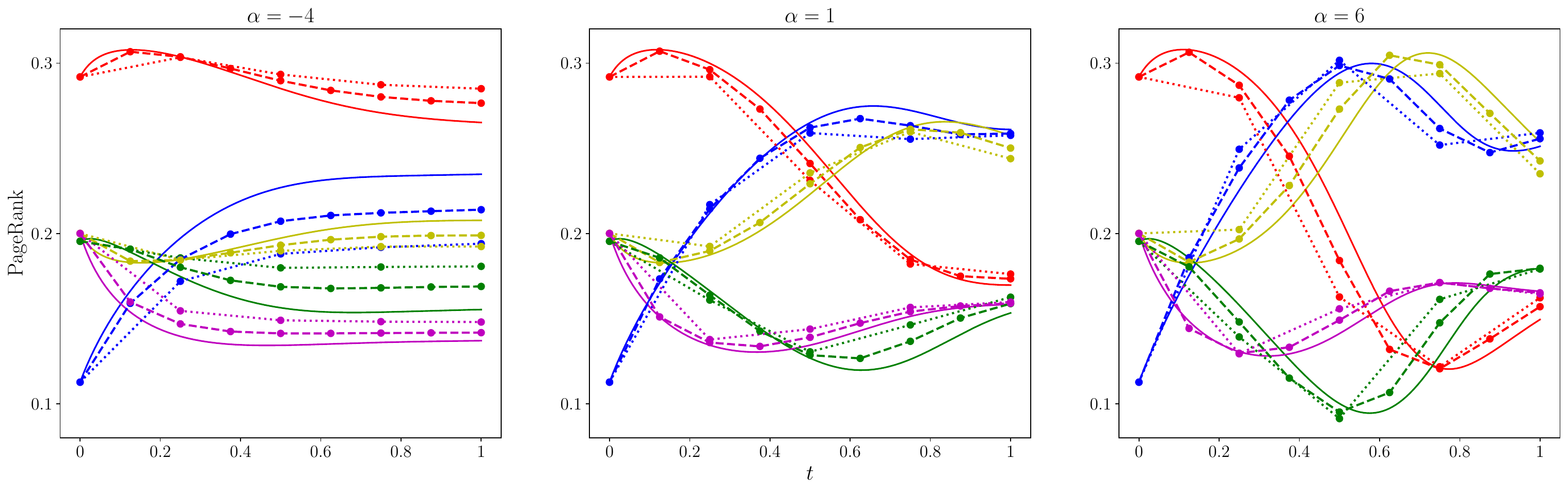}
    \caption{{\bf  The discrete PageRank (Definition \ref{def:local}) approximates to the continuous PageRank (Definition \ref{def:cont-local}).} PageRank values of five nodes (a color for each node) of a synthetic graph \eqref{synthetic} given by Definition \ref{def:cont-local} (continuous line) and by Definition \ref{def:local} with $N=5$ (dotted line) and with $N=9$ (dashed line) for the weights $\omega(s,t)=e^{-\alpha (t-s)}$ with $\alpha=-4$ (left panel), $\alpha=1$ (middle panel) and $\alpha=6$ (right panel).}
    \label{Fig3}
\end{figure*}

As an illustration we construct a non-directed continuous synthetic network with five nodes and the following links over the time-interval $I=[0,1]$, 
\begin{alignat}{2}
\label{synthetic}
&a_{12}=\frac 12 (\sin(2\pi t)+1), \quad  && a_{35}=0.5, \cr  
&a_{34}=\frac 1e (e^t-1),          \quad  && a_{25}=t^2, \cr
&a_{14}=\frac 12(\cos(2\pi t)+1),  \quad  && a_{23}=1-(t-1)^2.
\end{alignat}
Notice the absence of dangling nodes in this network, which is compatible with Definition~\ref{def:cont-local}. We will also consider the uniform personalization vector  $\mathbf{e}=(1,1, \dots, 1)$, the usual value $\lambda=0.85$ and the weights $\omega(s,t)=e^{-\alpha (t-s)}$ and  $\omega_{lk}=\omega(t_l,t_k)$. Next we compute both the continuous  (Definition~\ref{def:cont-local}) and the discrete (Definition~\ref{def:local}) PageRank  for two uniform discretizacions on $[0,1]$ with $N=5$ and $N=9$, respectively.  In Figure~\ref{Fig3} the continuous line represents the continuous PageRank, the dotted line represents the discrete PageRank on $I_5$ while the dashed line is the discrete PageRank on $I_9$; each colour represents each of the five nodes and three values of $\alpha$ are given, namely $\alpha=-4$ (left panel), $\alpha=1$ (middle panel) and $\alpha=6$ (right panel). As expected  the discrete approximations are better with increasing $N$. Additionally, if the  PageRank is computed on $I_{101}$, the difference with respect to the continuous  PageRank  reaches its maximum until $3 \cdot 10^{-3}$ units.    

We conclude with the announced localization result in the continuous setting. The proof is similar to the one given in the discrete setting. Notice that we might as well proceed by using the discrete result together with a dominated convergence theorem plus sandwich argument (Ref.~[\onlinecite{Cohn}]).

\begin{theorem}
With the notation above, we have, for every $t\in I$ and every $i\in\{1,\dots, n\}$  
\[
\min_{j} (X(t))_{ji}\le\  ({ \pi}^T(t))_i\ \le \max_j (X(t))_{ji}=(X(t))_{ii},
\]
where $X(t)=(1-\lambda)\left(Id-\lambda P_B(t)\right)^{-1}.$
\end{theorem}

\begin{acknowledgments}
This work has been partially supported by projects M2978 and M3033 (URJC Grants) and PID2019-107701G (Ministerio de Ciencia e Innovaci\'on Grant, Spanish Government).
\end{acknowledgments}
\vspace{5mm}

The data that support the findings of this study are openly available
in http://konect.cc/networks/, reference number  Ui (link-dynamic-itwiki).

\section*{References}
\bibliography{Temporalbib}

\end{document}